%% file: Cognitive_Energy_Harvesting.tex
\documentclass[conference]{IEEEtran} 
\usepackage{fancyhdr}
\usepackage{epsfig}
\usepackage{threeparttable}
\usepackage{epsf,epsfig}
\usepackage{amsthm}
\usepackage{amsmath}
\usepackage{amssymb}
\usepackage{amsfonts}
\usepackage[noadjust]{cite}
\usepackage{dsfont}
\usepackage{subfigure}
\usepackage{color}
\usepackage{caption}

\include{header}

\begin{document}
\title{\huge \setlength{\baselineskip}{30pt} Cognitive Energy Harvesting and Transmission from a Network Perspective}
\author{
\authorblockN{Seunghyun Lee}
\authorblockA{
National University of Singapore\\
Email: elelees@nus.edu.sg}
\and
\authorblockN{Kaibin Huang}
\authorblockA{
Hong Kong Polytechnic University\\
Email: huangkb@ieee.org}
\and
\authorblockN{Rui Zhang}
\authorblockA{
National University of Singapore\\
Email: elezhang@nus.edu.sg \vspace{-15pt}}} \maketitle 

\begin{abstract}
Wireless networks can be self-sustaining by harvesting energy from radio-frequency (RF) signals. Building on classic cognitive radio networks, we propose a novel method for network coexisting where mobiles from a secondary network, called  \emph{secondary transmitters} (STs),  either harvest energy from transmissions by nearby transmitters from a primary network, called  \emph{primary transmitters} (PTs), or transmit information if PTs are sufficiently  far away; STs store harvested energy in rechargeable batteries with finite  capacity and use all available energy for subsequent transmission when batteries are fully charged. 
In this model, each PT is centered at a  \emph{guard zone} and a \emph{harvesting zone} that are disks with given radiuses; a ST harvests  energy if it lies in some harvesting zone,   transmits fixed-power signals if it is outside all guard zones or else idles. Based on this model,  the spatial throughput  of the secondary network is maximized using a stochastic-geometry model where PTs and STs are modeled as independent homogeneous Poisson point processes (HPPPs), 
under  the outage constraints for coexisting networks and obtained in a simple closed-form. It is observed from the result  that the maximum secondary throughput  decreases linearly with the growing PT density, and the optimal ST density is inversely proportional  to  the derived transmission probability for STs. 
\end{abstract}

\section{Introduction} \label{Sec:Intro}

Powering mobile devices by harvesting ambient energy such as  solar energy, kinetic activities and electromagnetic radiation  makes mobile networks not only environmentally friendly but also self-sustaining. In this paper, we propose a novel method for network coexisting where mobiles in a cognitive-radio (secondary) network harvest radio-frequency (RF) energy from transmissions in a primary network besides opportunistically accessing the spectrum licensed to the primary network. The throughput of the secondary network is studied based on a stochastic-geometry model where the primary and secondary transmitters (PTs and STs) are modeled as independent homogeneous Poisson point processes (HPPPs). In this model, each PT is protected from the interference from STs by a \emph{guard zone} and delivers significant RF energy to STs located in a \emph{harvesting zone}. This model is applied to maximize the spatial throughput of the secondary network over the secondary transmission power and node density under outage-probability constraints for the primary and secondary networks.

Cognitive radios enable efficient spectrum usage by allowing a secondary network to share the spectrum licensed to a primary network 
\cite{Haykin:CognitiveRadio:2005}. Besides active development of  algorithms
for opportunistic transmissions by secondary users (see e.g., \cite{QingZhao:DynamicSpectrumAccess:2007, RuiZhang:DynamicResourceAllocation:2010}), research has been carried out on characterizing the throughput of coexisting networks based on stochastic-geometry network models.  The trade-off of the capacities between two coexisting networks \cite{Yin:ScalingLaws:2010}, \cite{ Huang:SpectrumSharing:2009} and among multiple networks \cite{Lee:SpectrumSharingTxCapacity:2011} has been studied. Moreover, outage probability of a Poisson-distributed cognitive radio network considering guard zones  has been analyzed in  \cite{Lee:PoissonCognitiveNetworks:2012}. 
Our work  also considers Poisson-distributed cognitive radio networks but differs from the prior studies in that the secondary users power their transmitters by harvesting energy from primary transmissions  instead of using traditional power supplies such as batteries. 

Recently, wireless communication powered by energy harvesting has emerged to be a new and active research area.  Given energy harvesting, the existing transmission algorithms need to be redesigned for different  wireless systems \cite{Ho&Zhang:OptimalEnergyAllocation:2012}, \cite{Ozel:EnergyHarvesting:2011} to account for the randomness and causality of energy arrivals. 
For broadcast channels, multi-antenna transmission for simultaneous wireless information and energy transfer has been studied in \cite{Yang:BroadcastingEnergyHarvesting:2012}. Also, the throughput of a MANET powered by energy harvesting has been investigated in \cite{Huang:ThroughputAdHocEnergyHarvesting:2012} where the network spatial throughput is maximized by optimizing transmission power under an outage constraint.  In contrast, our work studies a different type of  network, namely cognitive radio network powered by opportunistic RF-based energy harvesting \cite{Liu:OpportunisticEnergyHarvesting:2012}.

In this paper, it is assumed that time is slotted, and both the channels and HPPPs are quasi-static. The  harvesting zone and secondary transmission power are constrained to be sufficiently  small such that a ST lying in a harvesting zone is fully charged within a time slot. The main contributions of this work are summarized as follows. Based on the aforementioned  network model, we derive the transmission probability of STs using Markov-chain theory. Moreover, the maximum  network throughput of the secondary network under the outage  constraints is derived  by a Poisson approximation of the transmitting STs. It is shown that the maximum secondary throughput is linearly decreasing with the increasing PT density, and the optimal ST density is inversely proportional to the transmission probability of STs.

\section{Model and Metric} \label{Sec:SysMod}

\subsection{Network Model}

{\begin{figure} 
\centering
\includegraphics[width=7.5cm]{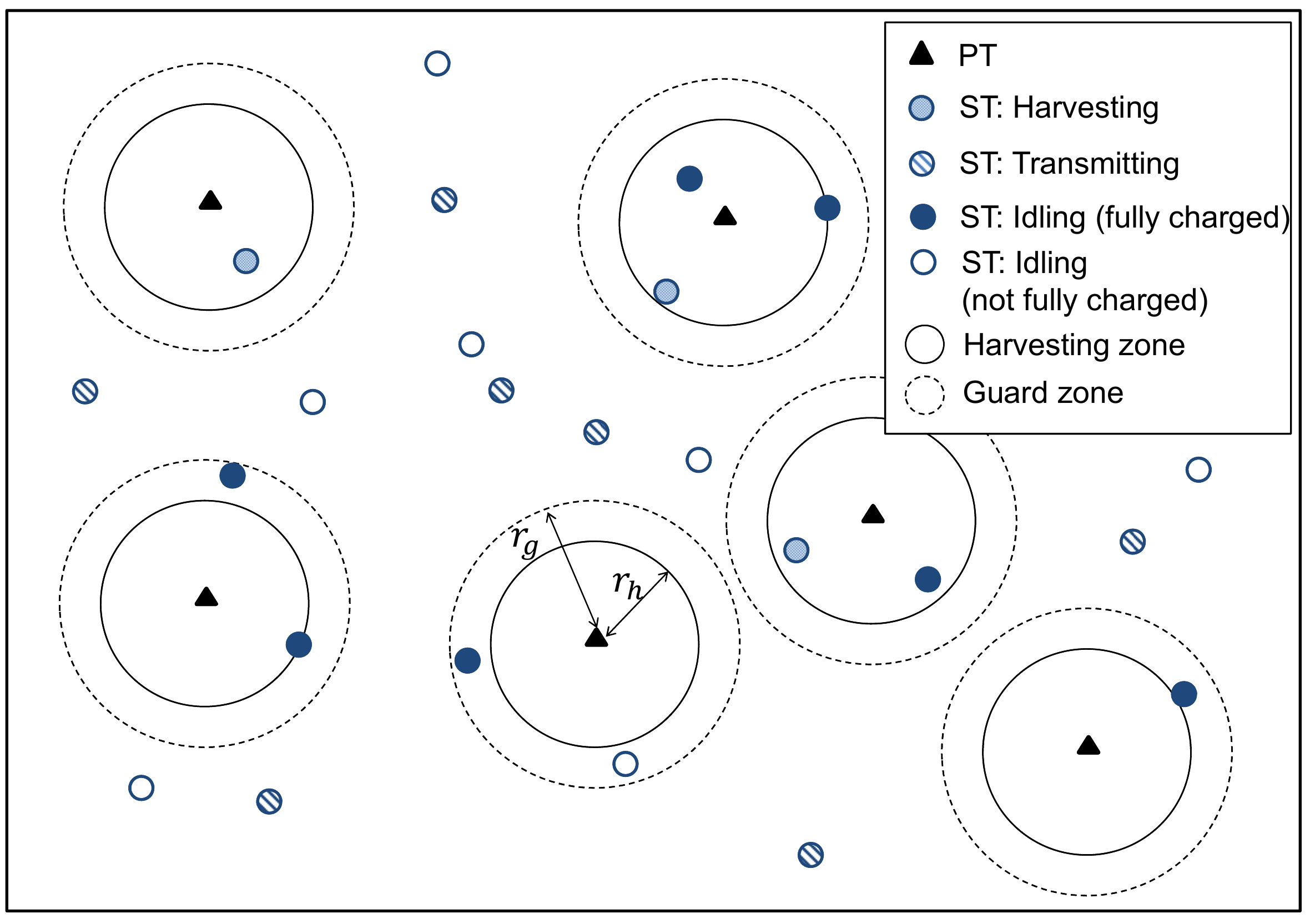}\vspace{1pt}
\caption{A stochastic-geometry model of the cognitive radio network where PTs and STs are distributed as independent  HPPPs. Each PT/ST has its intended receiver at an unit distance (not shown in the figure). A ST inside the harvesting zone harvests energy from the transmission by a nearby PT. To protect primary transmissions, all STs inside guard zones  are prohibited from transmission.}\vspace{-15pt}
\label{Fig:NetworkModel}
\end{figure} 

As shown in Fig.~\ref{Fig:NetworkModel}, the PTs and STs are distributed as independent HPPPs denoted by $\Phi_p = \{X\}$ and $\Phi_s = \{Y\}$ with density $\lambda_p$ and $\lambda_s$, respectively, where $X,Y \in \mathbb{R}^2$ denote the coordinates of the PTs and STs, respectively. Each transmitter transmits data with fixed power to an intended receiver at an unit distance. Time is partitioned into slots with unit duration. The point processes are fixed within each slot and vary independently over different slots. STs are prevented from transmission when they lie in the guard zones to protect primary transmissions. Specifically, consider a disk with radius $r_g$ centered at each $X$ and let
\begin{equation}
\mathcal{G} = \bigcup_{X \in \Phi_p}b(X,r_g)
\end{equation}
where $b(x,r)\subset\mathds{R}^2$ represents a disk of radius $r$ centered at $x$ and hence $b(X, r_g)$ is the guard zone with radius $r_g$ for protecting $X \in \Phi_p$. A ST $Y \in \Phi_s$ cannot transmit when $Y \in \mathcal{G}$. Note that intuitively a guard zone should be centered at  a primary receiver but is defined as above to simplify analysis \cite{QingZhao:DynamicSpectrumAccess:2007}. Let $R$ denote the distance between a typical ST, denoted as $Y'$, and its nearest PT. The probability $p_g$ that $Y'$ lies in a  guard zone  is given as \cite{PoissonProcesses}
\begin{align}
p_g &= \Pr(Y' \in \mathcal{G})  \label{Eq:GuardZoneDef}\\
&= \Pr(R \leq r_g) \label{Eq:GuardZoneRe}\\
&=1-e^{-\pi  \lambda_p r_g^2}. \label{Eq:GuardZoneRe:a}
\end{align}
 
\subsection{Energy-Harvesting Model}

Each ST $Y \in \Phi_s$ harvests energy from transmission by the nearest PT when it is located inside a \emph{harvesting zone}. Let $b(X, r_h)$ with $0<r_h <r_g$ represent the harvesting zone centered at $X\in\Phi_p$ such that a ST $Y \in b(X, r_h)$ harvests energy from the PT $X$. Channels are assumed to have no fading for simplicity. Let $P_p$ and $P_s$ denote the fixed  transmission power of PTs and STs, respectively. It is assumed that $P_p \gg P_s$.  Given $R\leq r_h$, the harvesting power of $Y'$ is $\eta P_p R^{-\alpha}$ where $0<\eta\leq 1$ denotes the harvesting efficiency and $\alpha >2$ is the path-loss exponent. Similarly as \eqref{Eq:GuardZoneRe:a}, the probability $p_h$ that a typical ST lies inside a harvesting zone is
\begin{equation}
p_h = 1-e^{-\pi  \lambda_p r_h^2}.
\end{equation} 
  We assume that each ST has finite battery capacity $P_s$ that is the same as its transmission power. Upon the battery being fully charged,  a ST transmits in next slot if it is outside all guard zones. 

\subsection{Performance Metric}
 Let $p_t$ denote the transmission probability of a typical ST $Y'$. Assuming interference-limited networks, the received signal-to-interference-ratio (SIR) is required to exceed a  target SIR for reliable transmission. Let $\theta_p$ and $\theta_s$ be the target SIRs for the typical primary and secondary receivers, respectively. The outage probability  is then defined as $\Pout^{(p)}=\Pr(\SIR^{(p)}<\theta_p)$ for the primary network  and  $\Pout^{(s)}=\Pr(\SIR^{(s)}<\theta_s)$ for the secondary network. Outage-probability constraints are applied such that  $\Pout^{(p)}\leq \epsilon_p$ and $\Pout^{(s)}\leq \epsilon_s$ with $0 < \epsilon_p, \epsilon_s <1$. Given fixed PT density $\lambda_p$ and transmission power $P_p$, the performance metric is  the spatial network-throughput $\mathcal{C}_s$ (bit/s/Hz/unit-area) of the secondary network given by
\begin{equation} 
\mathcal{C}_s = p_t \lambda_s \log_2(1+\theta_s) \label{Eq:NetThroughput}
\end{equation} 
under the outage constraints.

\section{Transmission Probability of Secondary Transmitter} \label{Sec:TxProb}

In this section, the transmission probability of a typical ST is analyzed. Note that the  minimum power harvested by a ST is $\eta P_p r_h^{-\alpha}$. Therefore, the battery of an energy-harvesting ST is fully charged within a slot if  $0< P_s \leq \eta P_pr_h^{-\alpha}$ and this scenario  is assumed in this paper.
 It then follows that the battery power level can only be 0 or $P_s$. Consider the finite-state Markov chain with the state space $\{0,1\}$ and let the state $0$ and $1$ denote the battery level of power $0$ and $P_s$, respectively. Furthermore, let $\mathbf{P}$ represent the transition-probability matrix that can be obtained as
\begin{equation}
\mathbf{P}=\l[\begin{array}{cc}
1-p_h & p_h \\ 1-p_g & p_g  \end{array}\r]
\end{equation}
with $p_g$ and $p_h$ given in the preceding section. 
The steady-state probability is then given by the vector $\boldsymbol{\pi}=[\pi_0\,\,\,\pi_1]$ using Markov-chain theory as shown in the proof of  the following proposition.  Note that $\pi_1$ represents the probability that the battery is fully charged. 


\begin{proposition} \label{Prop:TxProb:Case1}
For $0<P_s \leq \eta P_p r_h^{-\alpha}$, the transmission probability of a typical ST is
\begin{equation}
p_t = \frac{p_h(1-p_g)}{p_h+1-p_g} \label{Eq:TxProb:Case1}.
\end{equation}
\end{proposition}

\begin{proof}
$\boldsymbol{\pi}$ is the left eigenvector of $\mathbf{P}$ corresponding to the unit eigenvalue such that
\begin{equation}
\boldsymbol{\pi}\mathbf{P}=\boldsymbol{\pi}.
\end{equation}
From this equation, the distribution of the battery level at a typical ST is obtained as 
\begin{equation}
\pi_0 = \frac{1-p_g}{p_h+1-p_g}, \quad \pi_1 = \frac{p_h}{p_h+1-p_g}.
\end{equation}
Consequently, the desired result in \eqref{Eq:TxProb:Case1} is obtained by multiplying $\pi_1$ with $1-p_g$.
\end{proof}
It can be observed from the above result that  $p_t$ is independent with   $P_s$, provided that $0<P_s \leq \eta P_p r_h^{-\alpha}$, since the battery of an energy-harvesting ST is fully charged  over one slot if it is inside a harvesting zone.

\section{Network Throughput Maximization} \label{Sec:NetThroughput}
In this section, the network throughput of the secondary network defined in  \eqref{Eq:NetThroughput} is investigated under  the outage  constraints.

\subsection{Outage Probability}

There are four types of interference in the network model. Let $I_{pp}$, $I_{ps}$, $I_{ss}$ and $I_{sp}$ denote the aggregate interference power  from all PTs to a typical  primary receiver/secondary receiver, and from all transmitting STs to a typical  secondary receiver/primary receiver, respectively. Then the outage probability for the typical primary receiver located at the origin can be written as
\begin{equation}
\Pout^{(p)}  = \Pr \l(\frac{P_p}{I_{pp} + I_{sp}} < \theta_p  \r),
\end{equation} 
where $I_{pp} = \sum_{X\in\Phi_p}P_p|X|^{-\alpha}$, and $|X| \in \mathbb{R}_+$ is the distance between $X$ and the  origin.
It is important to note that $I_{sp}$ is a summation over the process of transmitting STs that  is not necessarily a HPPP. As a result, the analysis of  the distribution of $I_{sp}$  is intractable. To overcome this difficulty, we approximate the transmitting ST process as a HPPP with density $p_t\lambda_s$ and assume that it is independent with $\Phi_p$. This assumption is validated by simulation in the sequel. It is worth mentioning that a similar approximation  is also used  in  \cite{Lee:PoissonCognitiveNetworks:2012}.  Based on this approximation, we obtain the following lemma.  
\begin{lemma} \label{Lemma:Approx:PoutPrimary}
By approximating  the process of transmitting STs as a HPPP with density $p_t\lambda_s$, the outage probability of the primary network is given as 
\begin{equation} \label{Approx:PoutPrimary}
\Pout^{(p)} \approx \Pr\l(\sum_{T \in \Lambda(\tau_p)} |T|^{-\alpha} > 1\r)
\end{equation}
where $\Lambda(y)$ and $T \in \mathbb{R}^2$ denote a HPPP with density $y$ and the coordinate of a node, respectively, and 
\begin{equation} \label{Eq:lambdastar1}
\tau_p = \theta_p^{\frac{2}{\alpha}}\l(p_t\lambda_s\l(\frac{P_s}{P_p}\r)^{\frac{2}{\alpha}} + \lambda_p\r). 
\end{equation}
\end{lemma}
\begin{proof}
Let $\Phi_t \subset \Phi_s$ denote the point process of transmitting STs. Then $I_{sp}$ can be written as $I_{sp} = \sum_{Y\in \Phi_t \backslash b(X^\star,r_g)}P_s|Y|^{-\alpha}$ where $X^\star$ is the PT corresponding to the typical  primary receiver at the origin. Let  $I'_{sp} = \sum_{T \in \Lambda(p_t\lambda_s)} P_s|T|^{-\alpha}$ and  we approximate $I_{sp}$ as $I'_{sp}$. Then, it follows that
\begin{align}
&\Pout^{(p)} \approx \Pr \l(\frac{P_p}{I_{pp} + I'_{sp}} < \theta_p  \r) \\
\label{Eq:Mapping1} & = \Pr\l( \sum_{T \in \Lambda(\lambda_1^{(p)})} |T|^{-\alpha} +  \sum_{T \in \Lambda(\lambda_2^{(p)})}|T|^{-\alpha} >1 \r) \\ 
\label{Eq:Superposition1}& = \Pr\l( \sum_{T \in \Lambda(\lambda_1^{(p)} + \lambda_2^{(p)})} |T|^{-\alpha} >1 \r),
\end{align}
where $\lambda_1^{(p)} = \theta_p^{2/\alpha}\lambda_p$, $\lambda_2^{(p)} =\l(\frac{\theta_p P_s}{P_p}\r)^{2/\alpha}p_t\lambda_s $, and hence $\lambda_1^{(p)} + \lambda_2^{(p)} = \tau_p$. Note that \eqref{Eq:Mapping1} is from the fact that the point process $a\Lambda(\lambda)$ with $a>0$ is a HPPP with density $\lambda/a^2$ which can be proved by the Mapping Theorem \cite{Huang:ThroughputAdHocEnergyHarvesting:2012}, and \eqref{Eq:Superposition1} is from the Superposition Theorem \cite{PoissonProcesses}.
\end{proof}
It is worth mentioning that $\Pout^{(p)}$ depends only on $\tau_p$ and is an increasing function of $\tau_p$.

For the secondary network, it is important to note that the outage probability for the typical secondary receiver located at the origin is implicitly conditioned on that the  typical ST corresponding to this receiver is active (transmitting). Consequently, the typical ST must be lying outside  guard zones  and hence there are no PTs inside the disk of radius $r_g$ centered at the typical  ST. Let the event of this condition be denoted by $\mathcal{E} = \{\Phi_p \cap b(Y^\star,r_g) = \emptyset\}$, where $Y^\star$ denote the transmitting ST corresponding to the receiver at the origin. Then $\Pout^{(s)}$ can be  written as
\begin{equation} \label{Eq:PoutSecondary}
\Pout^{(s)} = \Pr\l(\frac{P_s}{I_{ps}+I_{ss}} <\theta_s \l| \mathcal{E} \r.\r),
\end{equation}
where  $I_{ps} = I_{pp}$ and $I_{ss} = \sum_{Y\in\Phi_t}P_s|Y|^{-\alpha}$. From the law of total probability, 
\begin{equation} \label{Eq:PoutSecondary2}
\Pout^{(s)} = \frac{\Pr(\frac{P_s}{I_{ps}+I_{ss}}<\theta_s) - \Pr\l(\frac{P_s}{I_{ps}+I_{ss}} <\theta_s \l| \bar{\mathcal{E}} \r.\r)\Pr(\bar{\mathcal{E}})}{\Pr(\mathcal{E})}.
\end{equation}
Note that $\bar{\mathcal{E}} = \{\Phi_p \cap b(Y^\star,r_g) \neq \emptyset\}$. Similar to Lemma~\ref{Lemma:Approx:PoutPrimary}, $I_{sp}'$ is a good approximation of $I_{ss}$, and hence we obtain the following lemma.

\begin{lemma} \label{Lemma:Approx:PoutSecondary}
By approximating  the process of transmitting STs as a HPPP with density $p_t\lambda_s$, the outage probability of the secondary network can be approximated as
\begin{equation} \label{Approx:PoutSecondary}
\Pout^{(s)} \approx \frac{\Pr\l(\sum_{T \in \Lambda(\tau_s)} |T|^{-\alpha} > 1\r)-p_g}{1-p_g},
\end{equation}
where 
\begin{equation} \label{Eq:lambdastar2}
\tau_s = \theta_s^{\frac{2}{\alpha}}\l(p_t\lambda_s + \lambda_p\l(\frac{P_s}{P_p}\r)^{-\frac{2}{\alpha}}\r).
\end{equation}
\end{lemma}
\begin{proof}
Assuming that $I_{sp}'$ is a good approximation of $I_{ss}$, it then follows that 
\begin{align}
& \Pr(\frac{P_s}{I_{ps}+I_{ss}}<\theta_s) \approx \Pr \l(\frac{P_s}{I_{ps} + I'_{sp}} < \theta_s  \r) \\
\label{Eq:Superposition2}&= \Pr \l( \sum_{T\in\Lambda(\lambda_1^{(s)} + \lambda_2^{(s)})}|T|^{-\alpha} > 1\r), 
\end{align}
where $\lambda_1^{(s)} = \theta_s^{2/\alpha}p_t\lambda_s$, $\lambda_2^{(s)}=\l(\frac{\theta_s P_s}{P_p}\r)^{-2/\alpha}\lambda_p$, and thus $\lambda_1^{(s)} + \lambda_2^{(s)} = \tau_s$. Note that  \eqref{Eq:Superposition2} follows the same arguments as for \eqref{Eq:Mapping1} and \eqref{Eq:Superposition1}. Next, with the aforementioned assumption that $P_p \gg P_s$,
it is likely that interference from even only one PT inside $b(Y^\star, r_g)$ is sufficient to cause outage at the typical secondary receiver at the origin. Consequently, we have $ \Pr\l(\frac{P_s}{I_{ps}+I_{ss}} <\theta_s \l| \bar{\mathcal{E}} \r.\r) \approx 1$. Substituting this result, \eqref{Eq:Superposition2} and the fact that $\Pr(\mathcal{E})=e^{-\pi\lambda_p r_g^2}=1-p_g$ into \eqref{Eq:PoutSecondary2}  yields the result in  \eqref{Approx:PoutSecondary}.
\end{proof}

Note that similarly as $\Pout^{(p)}$, $\Pout^{(s)}$ is an increasing function of $\tau_s$.


\subsection{Network Throughput Maximization}
The throughput of the secondary network, defined in \eqref{Eq:NetThroughput}, is maximized over  $P_s$ and $\lambda_s$, for a given pair of $P_p$ and $\lambda_p$ and under primary and secondary outage-probability constraints. The optimization problem can be written as
\begin{align}
\mbox{(P1):}\quad \max_{P_s, \lambda_s} & \quad  p_t\lambda_s \mbox{log}_2(1+\theta_s)\\
\mbox{s.t.}& \quad \Pout^{(p)} \leq \epsilon_p \\
& \quad \Pout^{(s)} \leq \epsilon_s.
\end{align}
As stated in Section~\ref{Sec:TxProb}, we assume $0<P_s \leq \eta P_p r_h^{-\alpha}$ in this paper and thus the transmission probability $p_t$ is a constant for a given pair of $r_g$ and $r_h$. Therefore, the maximum network throughput is achieved by simply maximizing $\lambda_s$ under the outage constraints. To solve Problem (P1), we use the nominal node densities \cite{Huang:ThroughputAdHocEnergyHarvesting:2012}, denoted as $\mu_p$ and $\mu_s$,  for HPPPs $\Lambda(\mu_p)$ and $\Lambda(\mu_s)$, respectively, such that 
\begin{equation}
\Pr\l( \sum_{T \in \Lambda(\mu_p)} |T|^{-\alpha} >1 \r) = \epsilon_p,
\end{equation}
\begin{equation}
\Pr\l( \sum_{T \in \Lambda(\mu_s) } |T|^{-\alpha} >1 \r) = (1-p_g)\epsilon_s + p_g.
\end{equation}
It should be taken into account that $\mu_p$ and $\mu_s$ are monotone-increasing functions of $\epsilon_p$ and $\epsilon_s$, respectively, which in general have no closed-form expressions and can only be found by simulation.  With these parameters, the solution of Problem (P1) can be found as given in the following proposition.

\begin{proposition} \label{Prop:Opt_Throughput}
The maximum network throughput of the secondary network  is given by:
\begin{enumerate}

\item If $\eta r_h^{-\alpha} < \frac{\theta_s}{\theta_p}\l(\frac{\mu_s}{\mu_p}\r)^{-\frac{\alpha}{2}}$, 
\begin{equation}
\mathcal{C}_s^* = \l(\theta_s^{-\frac{2}{\alpha}}\mu_s - \eta^{-\frac{2}{\alpha}}r_h^2\lambda_p\r)\log_2(1+\theta_s), \label{Eq:OptThroughput1}
\end{equation}
where the optimal ST transmit power is
\begin{equation} \label{Eq:OptTxPower1}
P_s^* = \eta P_p r_h^{-\alpha},
\end{equation}
and the optimal ST density is 
\begin{equation} \label{Eq:OptSTDen1}
\lambda_s^* = \frac{1}{p_t}\l(\theta_s^{-\frac{2}{\alpha}}\mu_s - \eta^{-\frac{2}{\alpha}}r_h^2\lambda_p\r).
\end{equation}

\item If $\eta r_h^{-\alpha} \geq \frac{\theta_s}{\theta_p}\l(\frac{\mu_s}{\mu_p}\r)^{-\frac{\alpha}{2}}$,
\begin{equation} \label{Eq:OptThroughput2}
\mathcal{C}_s^*= \l(\theta_s^{-\frac{2}{\alpha}}\mu_s - \l(\frac{\theta_s}{\theta_p}\r)^{-\frac{2}{\alpha}}\frac{\mu_s}{\mu_p}\lambda_p\r)\log_2(1+\theta_s),
\end{equation}
where the optimal ST transmit power is
\begin{equation} \label{Eq:OptTxPower2}
P_s^* = \frac{\theta_s}{\theta_p}\l(\frac{\mu_s}{\mu_p}\r)^{-\frac{\alpha}{2}}P_p,
\end{equation}
and the optimal ST density is
\begin{equation}
\lambda_s^* =  \frac{1}{p_t}\l(\theta_s^{-\frac{2}{\alpha}}\mu_s - \l(\frac{\theta_s}{\theta_p}\r)^{-\frac{2}{\alpha}}\frac{\mu_s}{\mu_p}\lambda_p\r).
\end{equation}
\end{enumerate}
\end{proposition}

\begin{figure}
\centering
\subfigure[$\eta r_h^{-\alpha}<\frac{\theta_s}{\theta_p}\l(\frac{\mu_s}{\mu_p}\r)^{-\frac{\alpha}{2}}$]{
\centering
\includegraphics[width=4.2cm]{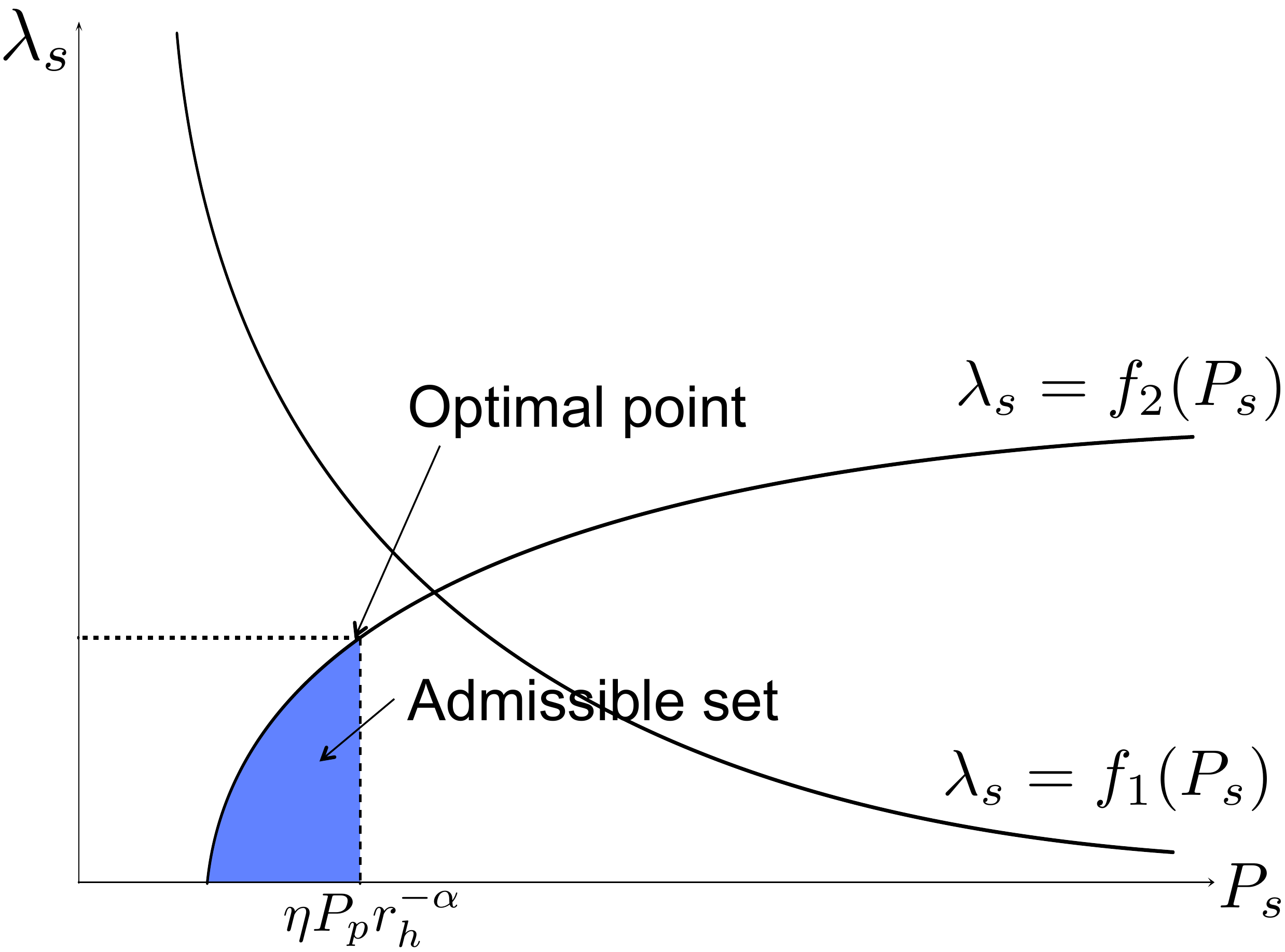}\vspace{5pt}\label{AdmissibleSet1}}
\subfigure[$\eta r_h^{-\alpha} \geq \frac{\theta_s}{\theta_p}\l(\frac{\mu_s}{\mu_p}\r)^{-\frac{\alpha}{2}}$]{
\centering
\includegraphics[width=4.2cm]{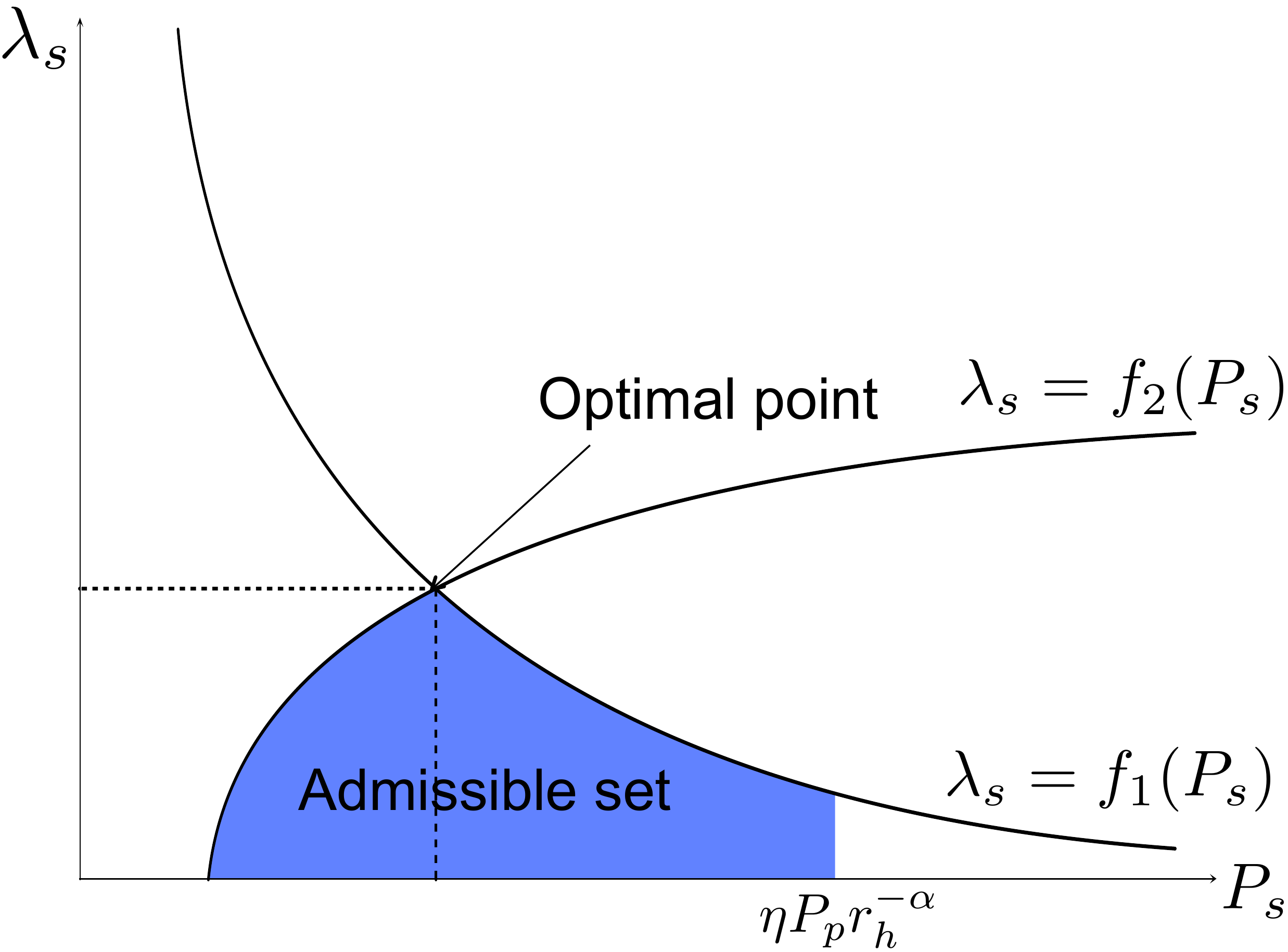}\vspace{5pt}\label{AdmissibleSet2}} \vspace{-5pt}
\caption{The admissible set of $(P_s,\lambda_s)$ shown as the shaded region for Problem (P1).}\vspace{-15pt} \label{Fig:AdmissibleSet}
\end{figure}

\begin{proof}
Since $\Pout^{(p)}$ and $\Pout^{(s)}$ are increasing functions of $\tau_p$ and $\tau_s$, respectively, Problem (P1) is equivalent to maximizing $\lambda_s$ subject to $\tau_p \leq \mu_p$, $\tau_s \leq \mu_s$ and $0< P_s \leq \eta P_p r_h^{-\alpha}$. In addition, from \eqref{Eq:lambdastar1} and \eqref{Eq:lambdastar2}, $\tau_p \leq \mu_p$ and $\tau_s \leq \mu_s$ are equivalent to $\lambda_s \leq f_1(P_s)$ and $\lambda_s \leq f_2(P_s)$, respectively, where
\begin{equation} \vspace{-5pt}
f_1(P_s) = \frac{1}{p_t}\l(\theta_p^{-\frac{2}{\alpha}}\mu_p - \lambda_p\r)\l(\frac{P_s}{P_p}\r)^{-\frac{2}{\alpha}},
\end{equation}
\begin{equation}
f_2(P_s) = \frac{1}{p_t}\l(\theta_s^{-\frac{2}{\alpha}}\mu_s - \lambda_p \l(\frac{P_s}{P_p}\r)^{-\frac{2}{\alpha}}\r).
\end{equation}
As plotted in Fig.~\ref{Fig:AdmissibleSet}, $f_1(P_s)$ decreases whereas $f_2(P_s)$ increases with growing $P_s$. The intersection point of the two curves $\lambda_s = f_1(P_s)$ and $\lambda_s = f_2(P_s)$ can be found by solving $f_1(P_s) = f_2(P_s)$, which yields $(P_s,\lambda_s) = \l(\frac{\theta_s}{\theta_p}\l(\frac{\mu_s}{\mu_p}\r)^{-\frac{\alpha}{2}}P_p, \frac{\mu_s(\mu_p-\theta_p^{\frac{2}{\alpha}}\lambda_p)}{p_t\theta_s^{\frac{2}{\alpha}}\mu_p}\r)$. The shaded region in Fig.~\ref{Fig:AdmissibleSet} shows the admissible set of $(P_s, \lambda_s)$ that satisfies $\lambda_s \leq f_1(P_s)$, $\lambda_s \leq f_2(P_s)$ and $0< P_s \leq \eta P_p r_h^{-\alpha}$. It can be observed from Fig.~\ref{AdmissibleSet1} that if  $\eta P_p r_h^{-\alpha} < \frac{\theta_s}{\theta_p}\l(\frac{\mu_s}{\mu_p}\r)^{-\frac{\alpha}{2}} P_p$, $\lambda_s$ attains its maximum when $P_s^*=\eta P_p r_h^{-\alpha}$. Thus, the results in \eqref{Eq:OptThroughput1} and $\eqref{Eq:OptSTDen1}$ are obtained from $\lambda_s^* = f_2(P_s^*)$. Otherwise, if $\eta P_p r_h^{-\alpha} \geq \frac{\theta_s}{\theta_p}\l(\frac{\mu_s}{\mu_p}\r)^{-\frac{\alpha}{2}}P_p$, it is observed from Fig.~\ref{AdmissibleSet2} that the optimal value of $\lambda_s$ is the intersection of $\lambda_s = f_1(P_s)$ and $\lambda_s = f_2(P_s)$, as obtained above.
\end{proof}

Note that if $\eta r_h^{-\alpha} < \frac{\theta_s}{\theta_p}\l(\frac{\mu_s}{\mu_p}\r)^{-\frac{\alpha}{2}}$, the secondary transmission power is small enough such that satisfying the secondary outage constraint guarantees  satisfying the primary outage constraint. Also note that $\mathcal{C}_s^*$ for the case of $\eta r_h^{-\alpha} \geq \frac{\theta_s}{\theta_p}\l(\frac{\mu_s}{\mu_p}\r)^{-\frac{\alpha}{2}}$ is always larger than that of $\eta r_h^{-\alpha} < \frac{\theta_s}{\theta_p}\l(\frac{\mu_s}{\mu_p}\r)^{-\frac{\alpha}{2}}$. 

Two remarks on the trade-offs between the two coexisting networks are in order.
\begin{enumerate}
\item The maximum secondary network throughput $\mathcal{C}_s^*$ decreases linearly with growing PT density $\lambda_p$. Also, the optimal ST density $\lambda_s^*$ is inversely proportional to the secondary transmission probability $p_t$.

\item If $\lambda_p$ goes to zero, $\mathcal{C}_s^*$ converges to its largest value; however, since $p_t$ converges to zero as $\lambda_p \rightarrow 0$, $\lambda_s^*$ diverges to infinity. It means that although sparse PT density results in larger secondary throughput, an extremely large number of STs  should be deployed to achieve the maximum value.   

\end{enumerate}

\section{Numerical Results} \label{Sec:Simulation}


Numerical results are provided in this section. The path-loss exponent is set as $\alpha=4$, the power ratio between the two networks is $\frac{P_s}{P_p} =0.1$, the harvesting efficiency is $\eta = 0.1$, and the radius of guard zone, harvesting zone is $r_g=2$, $r_h=1$.


\begin{figure}
\centering
\includegraphics[width=7cm]{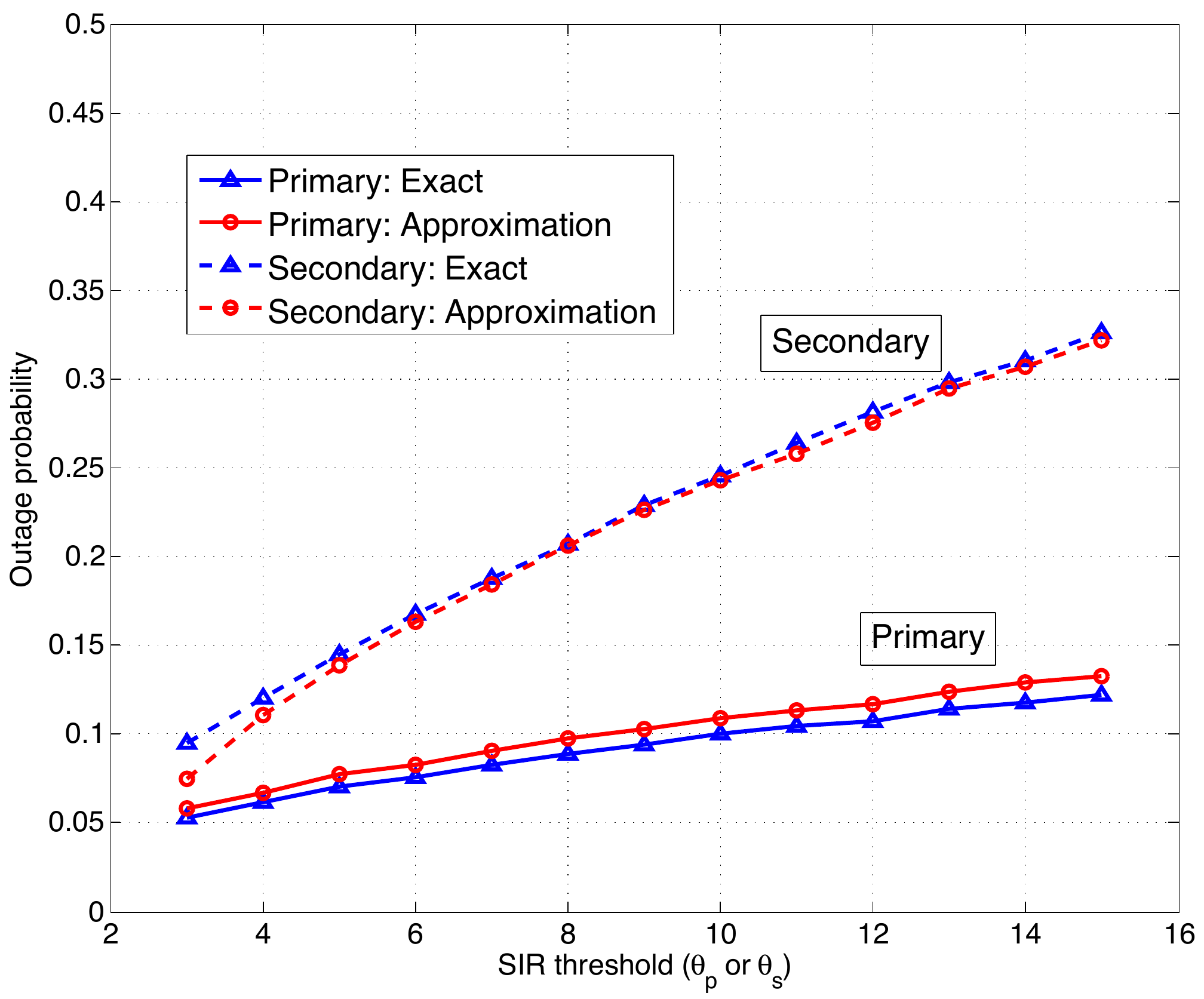}\vspace{-5pt}
\caption{ Comparison between the approximated and simulated outage probabilities with $\lambda_p=0.01$ and $\lambda_s = 0.1$.  \label{Fig:Outage} }\vspace{-15pt}
\end{figure}

Fig.~\ref{Fig:Outage} compares the outage probabilities obtained by simulation and those based on approximations in \eqref{Approx:PoutPrimary} and \eqref{Approx:PoutSecondary}. It is observed that our approximations are quite accurate.

\begin{figure}
\centering
\subfigure[Optimal ST density]{
\centering
\includegraphics[width=4.2cm]{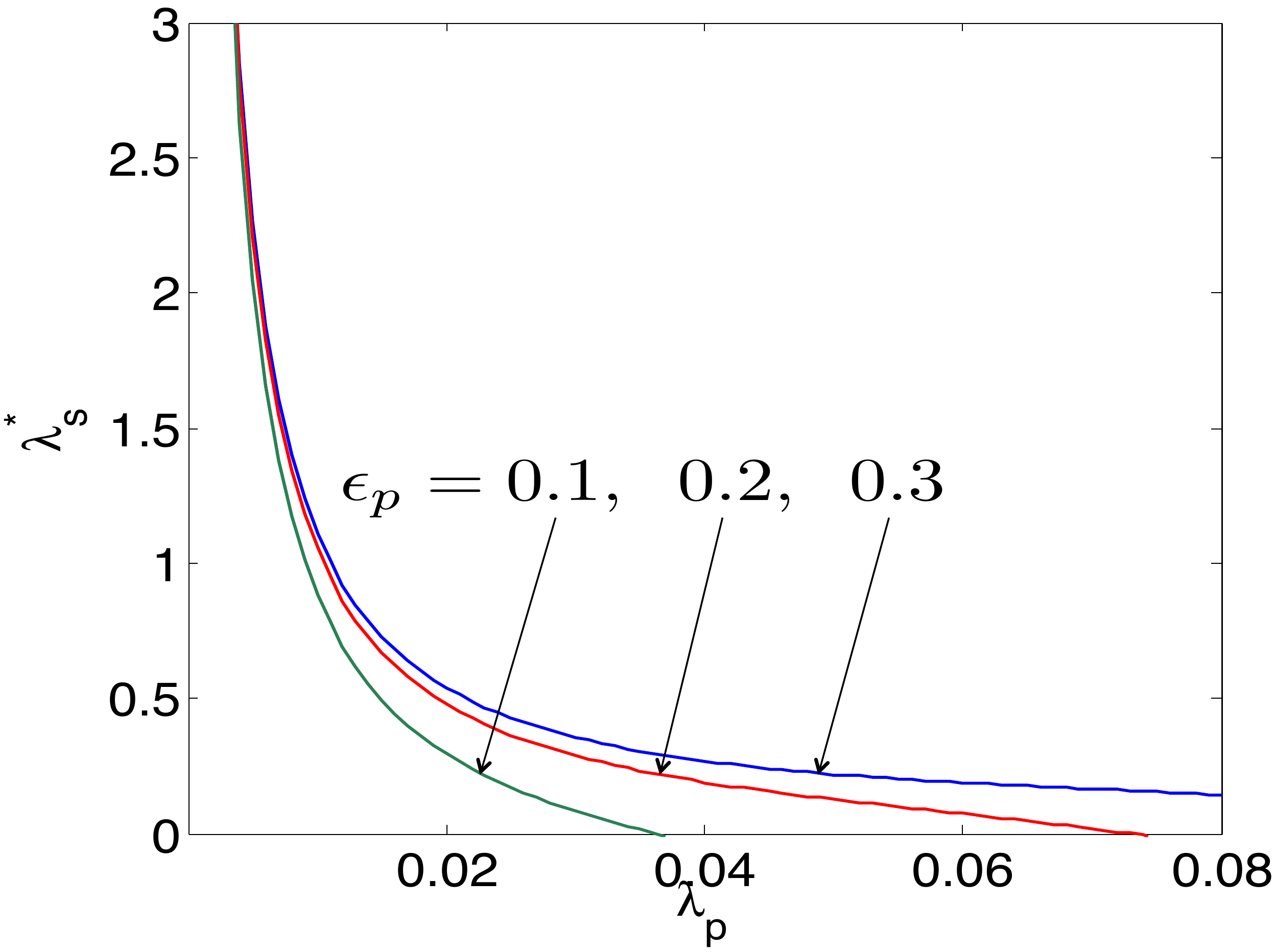}\vspace{3pt}\label{Fig:Opt_STDen}}
\subfigure[Maximum secondary throughput]{
\centering
\includegraphics[width=4.2cm]{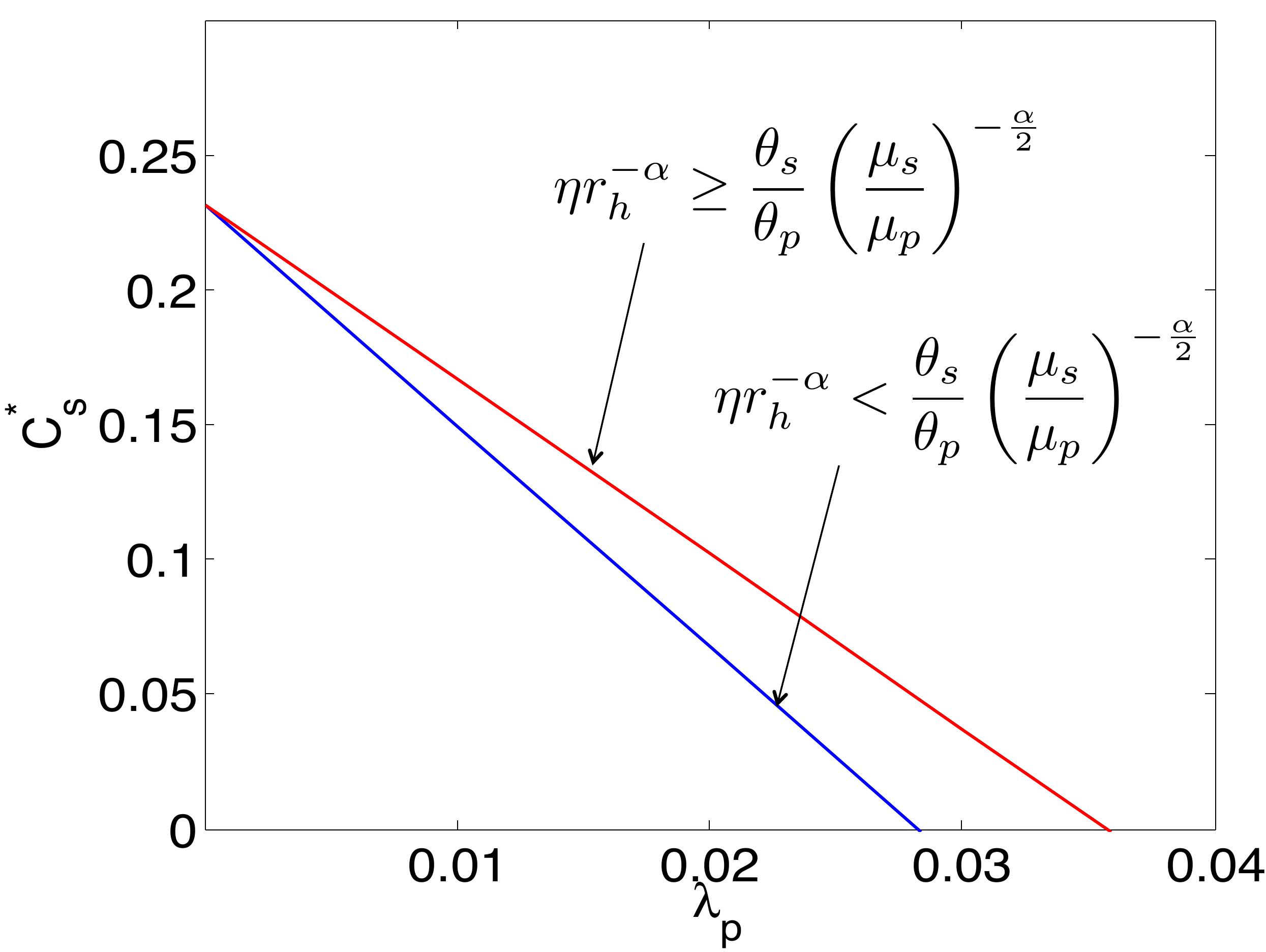}\vspace{3pt}\label{Fig:Opt_Throughput}} \vspace{-7pt}
\caption{Optimal ST density $\lambda_s^*$ and maximum secondary network throughput $\mathcal{C}_s^*$ versus PT density $\lambda_p$ with $\theta_p = \theta_s = 5$.} \label{Fig:Opt} \vspace{-15pt}
\end{figure}

The relation between $\lambda_s^*$ and $\lambda_p$ for the case of $\eta r_h^{-\alpha} \geq \frac{\theta_s}{\theta_p}\l(\frac{\mu_s}{\mu_p}\r)^{-\frac{\alpha}{2}}$ is plotted in Fig.~\ref{Fig:Opt_STDen} with three different values of   $\epsilon_p$ and $\epsilon_s = 0.4$. As mentioned above, $\lambda_p \rightarrow 0$ results in deploying an infinitely  large  number of STs to achieve the maximum secondary network throughput. Moreover, it is observed that larger $\epsilon_p$ allows deploying more STs as it permits more interference from the secondary network.  

The maximum secondary network throughput $\mathcal{C}_s^*$ versus $\lambda_p$ is plotted in Fig.~\ref{Fig:Opt_Throughput}. The parameters are set as $\epsilon_p = 0.35$ and $\epsilon_s = 0.4$ for the case of $\eta r_h^{-\alpha} < \frac{\theta_s}{\theta_p}\l(\frac{\mu_s}{\mu_p}\r)^{-\frac{\alpha}{2}}$, and $\epsilon_p = 0.2$ and $\epsilon_s = 0.4$ for the case of $\eta r_h^{-\alpha} \geq \frac{\theta_s}{\theta_p}\l(\frac{\mu_s}{\mu_p}\r)^{-\frac{\alpha}{2}}$. It is observed that $\mathcal{C}_s^*$ decreases linearly with increasing $\lambda_p$ .

\section{Conclusion} \label{Sec:Con}
In this paper, a novel scheme is proposed to enable the secondary network to harvest energy  and reuse the spectrum in the primary network. The secondary network throughput is maximized  and derived in a closed form based on a stochastic-geometry network model under outage constraints. 

\bibliographystyle{ieeetr}


\end{document}

%% file: header.tex
\newtheorem{lemma}{Lemma}

\newtheorem{proposition}{Proposition}

\def\phi{\varphi}

\def\SIR{\mathsf{SIR}}

\def\l{\left}
\def\r{\right}
\def\({\left(}
\def\){\right)}

\def\b0{{\mathbf{0}}}

\newcommand{\Pout}{P_{\mathsf{out}}}
